\documentclass{article}[11pt]

\usepackage[margin=1.0in]{geometry}

\usepackage{amsmath,amsthm,prodint}
\usepackage{forest} 
\usepackage{enumitem}
\usepackage{hyperref}       
\usepackage{url}            
\usepackage{booktabs}       
\usepackage{amsfonts}       
\usepackage{nicefrac}       
\usepackage{microtype}      
\usepackage{graphicx}
\usepackage{breqn}
\usepackage{natbib}
\usepackage{doi}
\usepackage{xcolor}
\usepackage{tikz-cd}
\usepackage{appendix}
\usepackage{comment}
\usepackage[capitalize,nameinlink]{cleveref}

\newtheorem{lem}{Lemma}
\newtheorem{cor}{Corollary}

\newtheorem{assp}{Assumption}
\newtheorem{thm}{Theorem}
\newtheorem{remark}{Remark}

\newcommand{\ind}{\perp\!\!\!\!\perp}
\newcommand{\M}{\mathcal{M}}
\newcommand{\full}{*}
\newcommand{\Pf}{\P^{\full}}

\newcommand{\E}{\mathbb{E}} 
\newcommand{\PP}{\mathbb{P}} 
\newcommand{\mO}{\mathcal{O}}
\renewcommand{\P}{\mathbb{P}} 

\renewcommand{\d}{\mathrm{d}}

\newcommand{\mC}{\mathcal{C}}

\begin{document}

\title{On Causal Inference for the Survivor Function}
\date{}
\author{Benjamin R. Baer \\
School of Mathematics and Statistics, University of St Andrews \\ St Andrews,  KY16 9SS, Scotland \\
https://orcid.org/0000-0002-9662-9377\\[.5ex]
Ashkan Ertefaie \\
Department of Biostatistics, Epidemiology \& Informatics\\
University of Pennsylvania, Philadelphia, Pennsylvania, U.S.A. \\
https://orcid.org/0000-0003-2611-9512 \\[.5ex]
Robert L. Strawderman \footnote{Corresponding author: \href{robert\_strawderman@urmc.rochester.edu}{robert\_strawderman@urmc.rochester.edu}
} \\ Department of Biostatistics and Computational Biology 
\\ University of Rochester, Rochester, New York, U.S.A. \\
https://orcid.org/0000-0002-6624-0272 \\ [.5ex]
}

\maketitle

\begin{abstract}
   In this expository paper, we consider the problem of causal inference and efficient estimation for the counterfactual survivor function. This problem has previously been considered in the literature in several papers, each relying on the imposition of conditions meant
   to identify the desired estimand from the observed data. These conditions, generally referred 
   to as either implying or satisfying coarsening at random, are inconsistently imposed across this literature 
   and, in all cases, fail to imply coarsening at random. We establish the first general characterization of
   coarsening at random, and also sequential coarsening at random, for this estimation problem.
Other contributions include the first general characterization of the set of all influence functions for the counterfactual survival probability under sequential coarsening at random, and the corresponding nonparametric efficient influence function. These characterizations are general in that neither impose continuity assumptions on either the underlying failure or censoring time distributions. We further show how the latter compares to alternative forms recently derived in the literature, including establishing the pointwise equivalence of the influence functions for our nonparametric efficient estimator and that recently given in \cite{westling2021inference}. \\[.5ex]
\end{abstract}

\newpage

\section{Introduction}

Survival analysis is a fundamental tool for the analysis of time-to-event data in various fields, such as medicine, epidemiology, and social sciences \citep{fleming1991counting, kalbfleisch2002statistical}. It involves studying the occurrence of events, such as death, disease progression, or failure, and understanding the underlying mechanisms that influence their timing. As researchers seek to go beyond descriptive analysis and understand the causal effects of interventions or treatments on survival outcomes, there is a growing interest in making valid causal inferences \citep{hernan2023what}.  

Of interest in this paper is causal inference for  the counterfactual failure survival function, a problem that has been previously studied by \citet{hubbard2000nonparametric, chen2001causal, anstrom2001utilizing, bai2013doubly, bai2017optimal, cai2020one, rytgaard2022continuous, westling2021inference}, among others. As part of our development, several inconsistencies in this literature pertaining to characterizations of ``coarsening at random'' (CAR) are carefully addressed and the first proper characterization of CAR and sequential CAR for causal survival analysis are rigorously developed. Under sequential CAR, we also characterize the class of all influence functions for the desired estimand as a class of augmented inverse probability weighted estimators, and determine the nonparametric efficient influence function. This efficiency theory is developed using approaches recently explained by \citet{hines2021demystifying} (see also \citealp{ kennedy2022semiparametric}), leading to estimators that are asymptotically efficient and robust to nuisance parameter estimation. Importantly, we do not make any absolute continuity assumptions, working with essentially arbitrary probability distributions for the observed data, as well as the required latent distributions for failure and censoring.  In related developments, we also refine results in earlier work where smaller classes of influence functions containing the efficient influence function have previously been proposed \citep[e.g.,][]{anstrom2001utilizing, bai2013doubly, bai2017optimal}. Finally, we prove by direct calculation that 
our nonparametric efficient influence function is pointwise equivalent to that recently derived in \citet{westling2021inference}.

In a companion paper, we carefully study the problem of causal inference for recurrent event processes when observation is terminated by death (i.e., failure, or a terminal event) \citep{BaerStrawErt25}. Many of the key technical results in this paper can actually be derived as special cases of the results in \citet{BaerStrawErt25}; we leverage these connections where appropriate to de-emphasize technical proofs here, instead focusing on enhancing understanding as to how previous developments are related to the current work.

\section{Notation: Counterfactual and Observable Random Variables}

Let $L$ denote a vector consisting of baseline variables for an individual. Assuming there is a binary treatment variable taking on the values zero and one, let $T^{\full}_0, T^{\full}_1$ be the potential failure times had an individual not been or been treated, respectively. We can then regard the unobservable full data for an individual as $(L, T^{\full}_0, T^{\full}_1).$ 

In causal inference for the counterfactual survival function, there are multiple coarsening variables involved. We denote the observed treatment (i.e., exposure) variable as $A,$ where $A = 1$ when treatment is assigned to an individual and $A = 0$ otherwise. We also define the potential coarsening variables $C^{\full}_0, C^{\full}_1,$ where $C^{\full}_a, a=0,1$ denotes the times that observation ceases had the exposure been $a$ (i.e., for reasons other than failure). In the present setting, we formulate the causal selection (i.e., involving $A$) as occurring before censoring (i.e., involving $C^{\full}_A$). 
The observed data $\mO$ is then defined by limiting the availability of the full data through the following mapping, 
\begin{equation*}
    \Phi(L,T^{\full}_0, T^{\full}_1; A, C^{\full}_0, C^{\full}_1) 
    = \Bigl( L, A, \Delta=I\{T^{\full}_A \leq C^{\full}_A\}, X=\min\{T^{\full}_A, C^{\full}_A\} \Bigr),
\end{equation*}
where the coarsening variables determine which part of the full data is available. As is evident in $\mO$, it is assumed
that the baseline variables $L$ and the exposure $A$ are observed, along with the failure indicator $\Delta$ and the observed time at risk $X$. Their definition incorporates the usual longitudinal assumption in survival analysis and the consistency assumption in causal inference. 
For later use, we also define the observed failure process $N_T(t) = I\{X \leq t, \Delta=1\}$ and the observed censoring process $N_C(t) = I\{X \leq t, \Delta=0\}$. These are one-jump processes that jump when a failure or censored time is observed, respectively.

Define the joint distribution of the full and coarsening variables as $\Pf,$ and let $\P$ as the distribution of the observed data induced by $\Pf$. Throughout, we use an asterisk over functions and random variables to denote reliance on the full or coarsened variables; a function or random variable without an asterisk only depends on the observed data. For $a=0,1$ we define $H^{\full}(u; a, l) = \Pf \{T^{\full}_a > u \mid L=l\}$ to be the conditional survival function for failure for the untreated and treated populations; similarly, let  $K^{\full}(u; a, l) = \Pf\{C^{\full}_a > u \mid L=l\}$ be the conditional censoring distribution and $\pi(a; l) = \PP\{A=a \mid L=l\}$ be the propensity score. Throughout, we adopt the convention that $0/0 = 0$.

\section{Counterfactual (i.e., Treatment-Specific) Survival Curves}
\label{sec:litreview:surv}

Define $\eta^*_a(t) = \PP^*\{T^*_a > t\}$ as the counterfactual survival function for failure. The study of this causal estimand has a reasonably long history; to the authors' knowledge, careful study began with \citet{hubbard2000nonparametric}, with increasing intensity of development over the past 25 years that includes, but is not limited to, \citet{van2000locally,chen2001causal,bai2013doubly,zhao2015doubly,rytgaard2022continuous,rytgaard2021estimation} and \citet{westling2021inference}.

\subsection{A proper characterization of coarsening at random (CAR)}
\label{sec:litreview:surv:car}

The following result gives the first proper characterization of CAR, and also sequential CAR, for the causal survival problem with
general failure and censoring distributions (i.e., under the indicated temporal ordering of coarsening of treatment selection, followed
by censoring).

\begin{thm}
\label{thm:car_caus_surv}
 The observed data mapping satisfies CAR if and only if for each $a=0,1$,
 \begin{align*}
     & (T^{\full}_0, T^{\full}_1) \ind A \mid L, \\
     & (T^{\full}_0, T^{\full}_1) \ind C^{\full}_a \mid A=a, L \text{ on } C^{\full}_a < T^{\full}_a, \\
     & T_{1-a}^{\full} \ind I(T_a^{\full} \leq C_a^{\full}) \mid A=a, L, T^{\full}_a.
 \end{align*}
 Additionally, the observed data mapping satisfies sequential CAR if and only if for each $a=0,1$,
    \begin{align*}
        & (T^{\full}_0, T^{\full}_1) \ind A \mid L, \\
        & T^{\full}_a \ind C^{\full}_a \mid A=a, L \text{ on } C^{\full}_a < T^{\full}_a. 
    \end{align*}
\end{thm}

A proof for the indicated characterizations of CAR and sequential CAR may be found in the Appendix.
Generally, CAR does not imply sequential CAR; the reverse statement is also true
\citep[e.g.,][]{gill1997sequential}.  

We now turn to several important papers in the area of causal survival analysis and consider the conditions
imposed for the purposes of identifiability and estimation. To our knowledge, \citet{hubbard2000nonparametric} 
is the first to attempt to characterize CAR for the causal survival problem, imposing the conditions
\begin{align*}
    & (T_0^*,T_1^*)  \ind A \mid L,\\
    & [[C^*_a=c \mid A=a,L,T^*_a]] = [[C^*_a=c \mid A=a,L]] \text{ if } c<T^*_a,
\end{align*}
where $[[\cdot]]$ denotes the Lebesgue hazard. However,
this characterization is evidently a special case of sequential CAR as given in the theorem above, and 
does not imply CAR in general. Parenthetically, 
each of \citet{moore2009increasing}, \citet{stitelman2010collaborative}, and \citet{cai2020one} also assert CAR, but fail to provide 
any explicit characterization of this condition. 
\citet{chen2001causal} instead assume that
    \begin{align*}
        & (T^{\full}_0, T^{\full}_1) \ind A \mid L, \\
        & (T^{\full}_0, T^{\full}_1) \ind C \mid A, L,
    \end{align*}
noting that the second condition further implies
that $T \bot C \mid A,L,$ where $T = A T^{\full}_1 + (1-A) T^{\full}_0.$
These authors do not make reference to CAR, per se, and also
do not appear to acknowledge the possibility of a counterfactual censoring time. However, if one assumes that their potential censoring time $C = A C^{\full}_1 + (1-A) C^{\full}_0$ and/or $C = C^*_0=C_1^*,$ then it can 
again be seen that their proposed conditions imply sequential CAR, but not CAR.

In later work, other conditions are imposed. For example,
in \citet{bai2013doubly}, it is assumed that
\begin{align*}
    & T^*_a \ind A \mid L, \\
    & C^*_a \ind T^*_a \mid A, L, \\
    & C^*_0=C_1^*
\end{align*}
for each $a=0,1.$ They further posit that these assumptions imply monotone CAR. Similarly, \citet{rytgaard2021estimation} 
assume, for each $a=0,1$, 
\begin{align*}
    & T^*_a \ind A \mid L, \\
    & C^*_a \ind T^*_a \mid A=a, L,
\end{align*}
and suggest that these conditions imply CAR. Finally,
in \citep{westling2021inference}, the conditions
\begin{align*}
    & T^*_a \ind A \mid L, \\
    & C^*_a \ind T^*_a \mid A, L, \\
    & C^*_a \ind A \mid L,
\end{align*}
for each $a=0,1$ are imposed. These conditions are similar
to those in \citet{bai2013doubly}, though no claim of
implying CAR or sequential CAR is made; we note that the
last condition serves a different purpose
and is not relevant to identification of the desired counterfactual 
survival curves.

In each case, the stated conditions neither lead to CAR nor sequential 
CAR as formulated in Theorem \ref{thm:car_caus_surv}. 
For example, in each of these papers, it can be seen that the second conditions are related to, 
but respectively stronger than, the corresponding condition for either CAR or sequential CAR in 
Theorem \ref{thm:car_caus_surv}. It can be also seen
that the conditions imposed in \citet{bai2013doubly} and
\citet{westling2021inference} imply the 
corresponding condition in \citet{rytgaard2021estimation}, but
additionally impose a cross-world independence assumption
in the case where $A \neq a$ \citep{andrews2020insights}. 
In contrast, the first condition in each of these papers is
weaker than the corresponding condition in Theorem \ref{thm:car_caus_surv} 
for both CAR and sequential CAR.  To see this, observe
that the conditional independence condition 
$(T^{\full}_0, T^{\full}_1) \ind A \mid L$ implies, 
for each relevant $t,u,$ that
\begin{equation}
\label{CAR cond}
\Pf\{ T^{\full}_a \leq t, T^{\full}_{1-a} \leq u, \mid A, L\}
= \Pf\{ T^{\full}_a \leq t, T^{\full}_{1-a} \leq u, \mid L\};
\end{equation}
letting $u \rightarrow \infty$ then shows
$
\Pf\{ T^{\full}_a \leq t \mid A, L\}
= \Pf\{ T^{\full}_a \leq t \mid L\}.
$
That is, assuming $(T^{\full}_0, T^{\full}_1) \ind A \mid L$ immediately implies the 
marginal independence conditions $T^{\full}_a \ind A  \mid L$ for $a=0,1$. However, 
because imposing $T^{\full}_a \ind A  \mid L$ for $a=0,1$ does not restrict the joint distribution of 
$(T^{\full}_0,T^{\full}_1),$ such marginal independence is necessarily 
weaker than the required condition under CAR (i.e., $(T^{\full}_0, T^{\full}_1) \ind A \mid L$) unless 
one further imposes the condition $T^{\full}_0 \ind T^{\full}_1 \mid A, L.$ To see that this latter
condition, combined with the marginal independence
condition $T^{\full}_a \ind A  \mid L,$ implies $(T^{\full}_0, T^{\full}_1) \ind A \mid L,$ 
we first note that $T^{\full}_a \ind A  \mid L$ implies 
\[
\Pf\{ T^{\full}_a \leq t \mid A, L\} 
\Pf\{ T^{\full}_{1-a} \leq u \mid A, L\} 
=
\Pf\{ T^{\full}_a \leq t \mid L\} 
\Pf\{ T^{\full}_{1-a} \leq u \mid L\}.
\]
However,
\[
\Pf\{ T^{\full}_a \leq t, T^{\full}_{1-a} \leq u \mid A, L\} =
\Pf\{ T^{\full}_a \leq t \mid A, L\} 
\Pf\{ T^{\full}_{1-a} \leq u \mid A, L\} 
\]
if and only if $T^{\full}_0 \ind T^{\full}_1 \mid A, L.$
Hence, in combination, 
the conditions $T^{\full}_0 \ind T^{\full}_1 \mid A, L$
and $T^{\full}_a \ind A  \mid L, a=0,1$ imply
\[
\Pf\{ T^{\full}_a \leq t, T^{\full}_{1-a} \leq u \mid A, L\}
= \Pf\{ T^{\full}_a \leq t \mid L\} 
\Pf\{ T^{\full}_{1-a} \leq u \mid L\}.
\]
Because the right-hand side is independent of $A,$
the left-hand side must also be independent of $A$;
that is, we have
$
\Pf\{ T^{\full}_a \leq t, T^{\full}_{1-a} \leq u \mid A, L\}
= \Pf\{ T^{\full}_a \leq t, T^{\full}_{1-a} \leq u \mid L\},
$
which is exactly \eqref{CAR cond}. 

In summary, to the authors' knowledge, Theorem \ref{thm:car_caus_surv} provides the first 
general characterization of CAR and sequential CAR in the causal survival analysis setting. As 
suggested by the general lack of equivalence of conditions that have been imposed
in the literature, neither CAR nor sequential CAR can be considered as necessary conditions for 
identifying $\eta^*_a$. Nevertheless, each remains conceptually important and provide sufficient 
conditions that will ensure the observed data likelihood factors into an ignorable term involving 
the coarsening variables and a term involving the full data \citep[e.g.,][]{cator2004testability}.

\begin{remark}
As noted above, the conditions imposed in \citet{bai2013doubly} neither imply CAR nor sequential CAR as formulated in 
Theorem \ref{thm:car_caus_surv}. Importantly, however, their conditions do imply monotone coarsening under a missing data model in which
$T^\full_1$ is assumed to be observed when $A = 1$ and $\Delta = 1$ and is considered to be missing otherwise. This approach, 
focused on estimating a functional of $\Pf\{T^\full_1 \leq t \},$ ignores the presence of the counterfactual $T^\full_0.$ As
will be seen later, this missing data view, although not truly causal, still leads to a valid but restricted class of influence functions for estimating $\Pf\{T^\full_1 \leq t \}.$ 
\end{remark}

\subsection{Observed Data Identification}
\label{sec:estimand:identif}

With the definitions of CAR and sequential CAR from  Theorem \ref{thm:car_caus_surv} in place, we turn to the question
of observed data identification. Unless otherwise noted, all results that follow will be derived by imposing the following assumptions. In particular,
$\PP^*,$ $\PP$ and the relationship between them are assumed to satisfy the following set of assumptions.
\begin{assp}
\label{asp:finitevar}
 All random variables have finite variance.
\end{assp}
\begin{assp}
\label{asp:phi}
 The map $\Phi$ correctly specifies the relationship between the full and observed data.
\end{assp}
\begin{assp}
\label{asp:car}
 Sequential coarsening at random holds.
\end{assp}
\begin{assp}
\label{asp:trunc}
 There exists a cutoff $\tau < \infty$ such that $X \leq \tau$ almost surely and $t \leq \tau$.
\end{assp}
\begin{assp}
\label{asp:positivity}
 There exists $\epsilon>0$ so that $\epsilon < \pi(a; l) K^{\full}(\tau; a, l)$ almost surely for all supported $a,l$.
\end{assp}

Assumption \ref{asp:finitevar} in standard in efficiency theory \citep{bickel1993efficient}.  
Assumption \ref{asp:phi} captures the consistency assumption in causal inference that one of the counterfactuals is observed and the ``longitudinal assumption'' in survival analysis that the time at risk is the minimum of the potential failure and censoring times. 
Assumption \ref{asp:car} states that sequential coarsening at random holds and hence that the two conditions in 
Theorem \ref{thm:car_caus_surv} hold. 
Assumption \ref{asp:trunc} is a technical condition that truncates the the analysis to a closed interval of finite length; this greatly simplifies the
technical analysis \citep{gill1983large}.  Together with Assumption \ref{asp:positivity}, Assumption \ref{asp:trunc} also implies that the potential censoring distribution has support that exceeds the support of the potential failure distribution. We note that Assumption \ref{asp:trunc} and Assumption \ref{asp:positivity} can be forced to hold by redefining the potential failure time to be truncated at some sufficiently small $\tau>0;$ a related truncation condition used to ensure positivity is 
often imposed in work involving inverse probability censoring weighted estimators.

Under Assumption \ref{asp:car}, the term  
\begin{align*}
 \pi(a; l) K^{\full}(\tau; a, l) 
 = \Pf \bigl \{ A=a,C^{\full}_a>\tau \mid L=l, T^{\full}_0, T^{\full}_1 \bigr\}
\end{align*}
appearing in Assumption \ref{asp:positivity} is simply the coarsening probability given the full data.
This expression reinforces the important point that the causal survival analysis problem has 
more than one coarsening variable; see also \citet{westling2021inference}.

Finall, define the model $\M^{\full}$ as the set of full and coarsened variable distributions satisfying all of the assumptions,  and define $\M$ as the derived set of observed data $\mO=(L,A,X,\Delta)$ distributions.

\subsubsection{Main Identification Result}
\label{sec:estimand:identif:results}

Define the usual at-risk process $Y(u) = I\{X \geq u\}$. 
Define $H(t; a, l)=\prodi_{u\in (0,t]} \{ 1 - \mathrm{d} \Lambda_T(u; a, l) \}$ for 
\begin{align*}
  \Lambda_T(t; a, l) = \int_{(0, t]} \frac{\d \E \{ N_T(u) \mid A=a, L=l\}}{\E \{ Y(u) \mid A=a, L=l\}},
\end{align*} 
where $\prodi$ is the product integral \citep{gill1990survey}. 
Similarly define $K(t; a, l) = \prodi_{u\in (0,t]} \{ 1 - \mathrm{d} \Lambda_C(u; a, l) \}$ for 
\begin{align*}
    \Lambda_C(t; a, l) 
    = \int_{(0, t]} \frac{\d \E \{ N_C(u) \mid A=a, L=l \}}{\E\{ Y^{\dagger}(u) \mid A=a, L=l\}},
\end{align*} 
where $Y^{\dagger}(u) = I\{X > u, \Delta=1 \text{ or } X \geq u, \Delta=0\}.$
In these expressions, $\E$ is taken under the observed data distribution
$\PP.$ The following lemma shows that $K$ identifies $K^{\full}$ and likewise that $H$ identifies $H^{\full}.$ 
\begin{lem}
\label{lem:surv-ident}
 Let $\Pf$ be any full and coarsened data distribution with corresponding observed data distribution $\P \in \mathcal{M}$. 
 Then $K^{\full}(u; a, l) = K(u; a, l)$ and $H^{\full}(u; a, l) = H(u; a, l)$ for all $u,a,l$.
\end{lem}
This lemma, proved in the Appendix, only assumes sequential CAR and does not impose the usual assumption of full conditional independence between the potential censoring and failure times. Additionally, it relies on a carefully constructed at-risk process $Y^{\dagger}$ for censoring that avoids the necessity of imposing absolute continuity assumptions on either failure or censoring. As such, and in contrast to much of the available literature,
we allow for the possibility that ties can occur between failure and censoring times.

\begin{remark}
    The process $Y^{\dagger}(\cdot)$ is a modified version of the usual at-risk process $Y(\cdot).$
Although $Y^{\dagger}(\cdot)$ may seem unfamiliar, it is in fact appropriate to use when events for which $\Delta=1$ have priority over events for which $\Delta=0$; see \citet[Page 56]{gill1994lectures} and \citet{strawbaer24} for 
related discussion. For example, this is the case for the observed data since $\Delta_A = I\{T^{\full}_A \leq C^{\full}_A\}$ is defined with a ``$\leq$'' rather than a ``$<$''; that is, any tie between $T^{\full}_A$ and $C^{\full}_A$ results in 
the observation of $T_A,$ a failure. We note that 
$Y^{\dagger}(u) = Y(u) - \{ N_T(u)-N_T(u-) \}$ for every $u>0;$ these
are almost surely equal when there are no shared discontinuities
between the underlying failure and censoring distributions
(e.g., at least one distribution is continuous). 
\end{remark}

The full data estimand is identified in $\M$, as the following proposition shows. This makes it possible 
to estimate the full data estimand using the observed data and  inverse probability weighting; a proof
is given in the Appendix.
\begin{thm}
\label{prop:identif}
    Let $\Pf$ be any full and coarsened data distribution with corresponding observed data distribution $\P \in \mathcal{M}$. For each time $t>0$, the full data estimand component $\eta_a^{\full}(t) = \eta_a(t)$, where $\eta_a(t) = \E \{ \varphi_{\eta, a} (t; \mO; \P) \}$ for
    \begin{equation*}
        \varphi_{\eta, a} (t; \mO; \P)
        = \frac{I\{A=a\}}{\pi(A; L)} \frac{\Delta}{K(X-; A, L)} I\{X > t\}.
    \end{equation*}
\end{thm}   
\noindent We note that $\eta_a(t)$ denotes the observed data identification of $\eta^*_a(t)$.

\subsection{The class of influence functions}

\citet[][Thm.\ 3]{BaerStrawErt25} derived the class of influence functions 
for the counterfactual mean of a recurrent event process when observation is terminated by death 
(i.e., failure, or a terminal event). In the absence of a recurrent event process, the
identification conditions under which this class of influence functions is derived
correspond to the sequential CAR conditions summarized in Theorem \ref{thm:car_caus_surv}. It follows
that the class of influence functions for  $\eta^*_a(t),$ assuming that the conditional coarsening 
distributions are known, can be obtained as a special case; the relevant result is summarized below.

\begin{thm}
\label{eta-IF-class} Under Assumptions 1-5,
 the class of influence functions for $\eta_a(t) = \eta^*_a(t)$ is given by
    \begin{align*}
       \frac{I\{A=a\}}{\pi(a; L)} & \frac{\Delta}{K(X-; a, L)} I\{X > t\} - \eta_a(t)
      - \{I\{A=a\} - \pi(a; L)\} h_1(L)  \\
      & + \int_{(0,\infty)} h_2(u; A, L) \,\d M_{C}(u; A, L),
    \end{align*}
    where  $h_1$ and $h_2$ are arbitrary index functions and
    $M_{C}(u;A,L) = N_C(u) - \int_{(0,u]} Y^\dag(s) d \Lambda_{C}(s; A,L).$
\end{thm}

\begin{remark}
As shown in \citet[Lemma 1]{BaerStrawErt25}, the class of influence functions
obtained when nominally allowing for  the presence of  treatment-specific counterfactual
censoring variables $C^\full_0 \neq C^{\full}_1$ is exactly equivalent to that derived under 
the assumption that there is a single coarsening variable for censoring, that is, $C^\full = C^\full_0 = C^{\full}_1.$
This is because the observed data  cannot distinguish between the coarsening variables
$C^\full$ and $C^\full_A = A C^\full_1 + (1-A) C^\full_0$. 
\end{remark}

The class given in Theorem \ref{eta-IF-class} is wider than the classes given in earlier work. For example, although not necessarily explicit, \citet{anstrom2001utilizing, bai2013doubly, bai2017optimal} each assume absolute continuity of the potential failure and censoring time
distributions and, for arbitrary $h_1$ and $h_2,$ report the class to have elements
\begin{align}
\nonumber
    \frac{I\{A=a\}}{\pi(a; L)} & \frac{\Delta}{K(X-; a, L)} I(X > t) - \eta_a(t) - \frac{I\{A=a\} - \pi(a; L)}{\pi(a; L)} h_1(a,L) \\
\label{eq:tsiatisclass}
    & + \frac{I\{A=a\}}{\pi(a; L)} \int_{(0,\infty)} h_2(u; a, L) \frac{\d \tilde M_C(u; a, L)}{K(u; a, L)}
\end{align}
for $\tilde M_{C}(u;A,L) = N_C(u) - \int_{(0,u]} Y(u) d \tilde \Lambda_{C}(u; A,L),$ and where
$\tilde \Lambda_{C}(u; A,L)$
is defined analogously to
$\Lambda_{C}(u; A,L),$ 
with $\E\{ Y(u) \}$
replacing $\E\{ Y^\dag(u) \}.$
Under continuity of failure, censoring or both, the at-risk process
$Y^\dag(u) = Y(u)$ and hence $M_{C}(u;A,L) = \tilde M_{C}(u;A,L).$
As will be seen in the next section, both 
the influence function \eqref{eq:tsiatisclass}
and that given in Theorem \ref{eta-IF-class} include the efficient influence function (i.e., when absolute continuity holds); however, the class \eqref{eq:tsiatisclass} is considerably smaller. In particular, the augmentation term for censoring (i.e., the summand with $\tilde M_C$) vanishes when $A\neq a,$ while this does not occur for the class summarized in Theorem \ref{eta-IF-class}.

A detailed derivation of the class in \eqref{eq:tsiatisclass} is not explicitly provided in the papers cited above. In the Appendix, and following the brief description in \citet[\S 2.1]{bai2013doubly}, we show how their class can be derived under a monotone missing data model rather than a causal survival model that jointly considers all counterfactuals. In particular, fixing $a,$ it will be seen that it is sufficient to consider $T_a$ as being observed when $A = a$ and
$\Delta = 1,$ and labeled as missing otherwise. In view of such a derivation, it should not come as a surprise that the class of influence functions in \eqref{eq:tsiatisclass} is smaller; the focus on estimating $\eta^*_a(t)$ for a given level of $a$ requires less data than that required under the indicated joint causal model considered in this paper.

\subsection{The form of the efficient influence function}
\label{sec:litreview:surv:eif}

To the authors' knowledge, the efficient influence function for $\eta_a(t)$ was first derived in \citet{hubbard2000nonparametric}; see also \citet{bai2013doubly}, who make use of developments in \citet{tsiatis2006semiparametric}. Both give the efficient influence function as a re-centered  augmented inverse probability weighted estimator, and each derive these results assuming absolute continuity of the latent failure and censoring distributions.  
Using results in \citet[][Thms.\ 3 \& 4]{BaerStrawErt25}, one may conjecture that the nonparametric efficient influence function 
within the class of influence functions described by Theorem \ref{eta-IF-class} is given by
   \begin{align}
\nonumber        \varphi_{\eta, a}(t, \mO; \P) - & \eta_a(t)
            - \frac{I\{A=a\} - \pi(a; L)}{\pi(a; L)} H(t; a, L)  \\
   \label{eq:EIF}
          &
            + \frac{I\{A=a\}}{\pi(a; L)} \int_{u\in(0,\infty)} \frac{H(t \vee u; a, L)}{H(u; a, L)} \frac{\d M_C(u; a, L)}{K(u; a, L)}.
    \end{align}
It is easy to show that \eqref{eq:EIF} reduces to the efficient influence function for $\eta_a(t)$ in both \citet{hubbard2000nonparametric} and \citet{bai2013doubly} when the indicated continuity conditions hold on the latent failure and censoring time 
distributions. However, this does not necessarily prove that \eqref{eq:EIF} is the nonparametric efficient influence
function under general conditions on these latent distributions.

\citet{stitelman2010collaborative, moore2009increasing, zhao2015doubly, diaz2019improved, westling2021inference, rytgaard2021estimation} 
each obtain an alternative form of the efficient influence function for $\eta_a(t)$, given below as
\begin{equation}
\label{eq:surv-gcomp-eif}
    H(t; a, L) 
    - \eta_{a}(t)
    - \frac{I\{A=a\}}{\pi(a; L)} \int_{(0,t]} \frac{H(t; a, L)}{H(u; a, L)} 
    \frac{\d M_T(u; a, L) }{K(u-; a, L)},
\end{equation}
where $M_T(u; a,L) = N_T(u) - \int_0^u Y(s) \,\d \Lambda_{T}(s; a, L)$
is the usual failure time martingale process (i.e., with respect
to the usual observed data filtration generated by 
$N_T,$ $Y,$ $A$ and $L$). The following result, proved in the Appendix, shows that 
\eqref{eq:EIF} and \eqref{eq:surv-gcomp-eif}, are pointwise equal, and hence
that \eqref{eq:EIF} is the nonparametric efficient influence function.

\begin{thm}
\label{thm: EIF equal}
The nonparametric efficient influence
functions \eqref{eq:EIF}
and \eqref{eq:surv-gcomp-eif}, are pointwise equal
in $t$ for every $t \leq \tau.$
\end{thm}

We highlight here that the usual at-risk process $Y$ appears in the failure martingale in the integrator of the augmentation term 
in \eqref{eq:surv-gcomp-eif}. Because this is equal to the efficient influence function \eqref{eq:EIF}, and this latter influence function involves a stochastic process for censoring $M_C$ defined using $Y^{\dagger}$, the resulting equivalence reinforces the fact that $Y^{\dagger}$ is the appropriate at-risk process for use in defining $M_C$ with general failure and censoring distributions.

However, even more can be said. First,
under continuity of the censoring mechanism, it is well
known that the stochastic process
$M_C(\cdot;a,L)$ is a martingale with respect to the same 
observed data filtration
as $M_T(\cdot; a,L).$ Consequently, in this case, the last term appearing on the right-hand side of \eqref{eq:EIF} is a mean zero martingale. Second, and more interesting, is that
$M_C(\cdot;a,L)$ is a mean zero martingale under the observed
data filtration in the absence of any continuity assumptions
on failure or censoring. Therefore,
in the presence of possible ties between failure and censoring,
the last term appearing on the right-hand
side of \eqref{eq:EIF} remains a mean-zero martingale
under the indicated observed data filtration. These results respectively follow from 
results in  \citet{baerstraw25}, and may be considered surprising in part
because $Y^\dag(u) = Y(u) - \{ N_T(u) - N_T(u-)\}$ 
is not a predictable process in the setting where ties can occur
between failure and censoring times.

\section{Estimation in Practice}
\label{sec:litreview:surv:other-ests}

A variety of approaches for estimation of the
counterfactual survival function $\eta^*_a(t) = \eta_a(t)$
have been proposed
to date. Earlier works, such as
\citet{hubbard2000nonparametric,  anstrom2001utilizing, bai2013doubly, bai2017optimal}, do not make use of machine learning or other 
fully nonparametric procedures to estimate unknown nuisance
parameters (i.e., the treatment propensity $\pi(a,L),$
survivor function $S(\cdot; a, L)$ or related functionals,
and/or the censoring survival function
$K(\cdot; a, L)$). The efficient
influence functions reported in these works
coincide with \eqref{eq:EIF} when all modeling
assumptions are correct, and typically lead to doubly robust
estimators otherwise. Several such estimators for $\eta_a(t)$ are described and compared by \citet{denz2022comparison}; these include estimators obtained under IPTW and AIPTW, G-computation, propensity score matching, empirical likelihoood, and 
those derived from pseudovalues. In general, similar
bias performance across such methods were observed when
the assumptions were met for consistency, with better
performance observed for methods making use
of ``outcome regression'' models (e.g., AIPTW).

In more recent work, including \citet{cai2020one, rytgaard2021estimation,westling2021inference, Wolock}, 
machine learning methods are leveraged to estimate all required nonparametric nuisance parameters,
leading to nonparametric efficient $\sqrt{n}-$asymptotic inference for the counterfactual survivor function $\eta^*_a(t)$
and relevant contrasts, such as $\eta^*_1(t)- \eta^*_0(t),$ when each required nuisance parameter is estimated
consistently at sufficiently fast rates (i.e., $o_P(n^{-1/4})$). Compared to others noted above, the methods in 
\citet{Wolock} have the advantage of being structured to allow for the use of the many different learning algorithms 
that have been developed for binary classification problems, and are also unique in 
the causal survival literature by allowing for left truncation in addition to right censoring. In particular, although
estimation in each case is ultimately  based on \eqref{eq:surv-gcomp-eif},
\citet{Wolock} demonstrates that their global stacking procedure provides some improvement over the
superLearner implementation described in \citet{westling2021inference}. 
Results summarized in \citet{BaerStrawErt25} show that estimation based on \eqref{eq:EIF},  
implemented using superLearner,  is in practice essentially indistinguishable from the 
approach described in \citet{westling2021inference} as long as the sample size is sufficiently large.

\section{Discussion}
\label{sec:conc}

The field of causal survival analysis continues to expand at a
rapid pace. The primary goal of the present work is to clarify the appropriate characterization of coarsening at random, and sequential CAR, in this setting. We have carefully compared these conditions to other identification assumptions that have previously been used in the literature, many of which have been erroneously suggested as implying CAR. We have also demonstrated how different classes
of influence functions have been derived, and established
how the different classes (and most efficient choices) are
related to each other. 

An open question of some interest is whether there are settings of practical importance for which imposing conditions that violate both CAR and sequential CAR
lead to consequential differences in estimators. For example, it is well known that characterizing CAR in multivariate right-censored data is challenging \citep{van1996efficient}. In such cases, researchers often impose the sequential randomization assumption (SRA), which is stronger than CAR. This implies that the model space under SRA is strictly smaller than the one under CAR and thus cannot be nonparametric; as a result, estimators derived under SRA are not nonparametrically efficient \citep[Section 5]{van2003unified}.

A more nuanced characterization of CAR in the presence of time-varying covariates and right censoring would also be valuable—particularly in settings where time-varying covariates, such as recurrent adverse events, are collected longitudinally. One compelling application would be the identification of mediated effects in such longitudinal survival data structures.

\section*{Acknowledgements}

Mark van der Laan provided helpful input on some of the content in \cref{sec:litreview:surv:car}. 
Butch Tsiatis provided helpful input on some of the content in \cref{sec:litreview:surv:class}. 
BRB thanks David Oakes for teaching him that, informally, censoring occurs later than failure. 
Parts of this work were supported by the National Institute of Neurological 
Disorders and Stroke (AE, BRB: R61/R33 NS120240) and the
National Institute of Environmental Health Sciences
(AE, RLS: R01 ES034021).

\bibliographystyle{apalike}
\bibliography{references} 

\newpage
\begin{appendices}
\renewcommand*{\thesubsection}{A\arabic{section}.\arabic{subsection}}
\renewcommand{\theequation}{A\arabic{equation}}
\renewcommand{\thefigure}{A\arabic{figure}}
\renewcommand{\thetable}{A\arabic{table}}
\renewcommand{\bibnumfmt}[1]{[A#1]}
\renewcommand{\citenumfont}[1]{A#1}

\renewcommand*{\thesection}{A\arabic{section}}

\section{Two Fundamental Identities }
\label{supp:sec:fund-ident-proofs}

In this section, we state and prove two 
general identities, both of which have their 
roots in the earliest work in coarsening theory.
In the first subsection below, we 
review two important preceding
results; in the second subsection, 
we state and prove two
important generalizations.

\subsection{Review}
\label{sec:alg-ident:intro}

In \citet[Section 3h]{robins1992recovery}, four fundamental identities are given. Two of these (i.e., equations 3.10a and 3.10b)  are the result of simple integral calculus; the remaining identities have led to important insights and developments in modern survival analysis. However, before stating these latter two identities, we need to introduce some notation. Define the censoring hazard $\lambda_C(t) = \lim_{h \to 0+} \P (X < t + h, \Delta=0 \mid X \geq t) / h$ 
and the survival function $K(t) = \exp \left\{ - \int_{(0,t]} \lambda_C(u) \,\d u \right\}.$  Then, as shown in 
\citet{robins1992recovery}, we can write
\begin{align}
    \frac{I(X \geq t)}{K(t)} 
    & \stackrel{a.s.}{=} \frac{\Delta}{K(X)} I(X \geq t) + \int_{(t,\infty)} \frac{\d N_C(u) - I(X \geq u) \lambda_C(u) \, \d u}{K(u)}, \nonumber \\ 
    \frac{\Delta}{K(X)} 
    & \stackrel{a.s.}{=} 1 - \int_{(0,\infty)} \frac{\d N_C(u) - I(X \geq u) \lambda_C(u) \,\d u}{K(u)}, \label{eq:rr-rob}
\end{align}
where ``$\stackrel{a.s.}{=}$'' denotes almost sure equality. 
These identities hold when the time at-risk $X$ is absolutely continuous (with respect to the Lebesgue measure). Under an analogous absolute continuity assumption, \citet{strawderman2000estimating} highlighted that \eqref{eq:rr-rob} actually holds algebraically; that is, the identities rely principally on the relationship between $\lambda_C(\cdot)$ and $K(\cdot),$ and not on any particular relationship between $K(\cdot),$ $X,$ and/or  $\Delta$.

\subsection{The identities}
\label{sec:alg-ident:idents}

In this subsection, we state and prove general versions of \eqref{eq:rr-rob} that 
hold exactly (i.e., not just almost surely), while also relaxing the absolute 
continuity assumption that underpins the identities summarized
in \eqref{eq:rr-rob}.

Let $X \in (0, \tau]$ and $\Delta\in\{0,1\}$ be arbitrary numbers (or random variables) 
that need not be related in any way. 
Let $K$ be right-continuous, have locally bounded variation on any finite interval, and satisfy $K(0)=1$ and $K(u) \neq 0$ for all $u \in (0,\tau],$
where $\tau < \infty.$ Here, $K$ is an arbitrary function that is not necessarily the survival function of some potential censoring time. 
Although we denote domains of integration as extending to $\infty$, they formally stop at $\tau$ due to the assumptions. Now, we can define
\begin{equation*}
\label{eq:lamc-defn}
    \Lambda_C(t) = - \int_{(0, t]} \frac{\d K(u)}{K(u-)};
\end{equation*}
note this is the cumulative hazard for censoring when $K$ is a survivor function for censoring, and is right-continuous in general. Define the function  $N_C(t) = I\{X \leq t, \Delta = 0\}$ and corresponding
``centered''  version
\begin{equation}
    M_C(t)
    = N_C(t) - \int_{(0,t]} Y^{\dagger}(u) \,\d \Lambda_C(u), \label{eq:mc}
\end{equation}
where  $Y^{\dagger}(u) = I(X > u, \Delta=1 \text{ or } X \geq u, \Delta=0).$ Then, we can prove the following result.

\begin{lem}
\label{lem:rr-lems}
Let $t>0$. 
Under the general setup just defined, the following two identities hold:
\begin{align}
    \frac{I(X > t)}{K(t)} 
    & = \frac{\Delta}{K(X-)} I(X > t) + \int_{(t,\infty)} \frac{\d M_C(u)}{K(u)} \label{eq:rr1} \\
    \frac{\Delta}{K(X-)} 
    & = 1 - \int_{(0,\infty)} \frac{\d M_C(u)}{K(u)}. \label{eq:rr2}
\end{align}
\end{lem}

\begin{proof}[Proof of Lemma \ref{lem:rr-lems}]
 We start by stating a few helpful expressions. First, 
 \begin{align*}
  \frac{1}{K(t)} - 1
  = - \int_{(0,t]} \frac{\d K(u)}{K(u-)K(u)}
  = \int_{(0,t]} \frac{\d \Lambda_C(u)}{K(u)},
 \end{align*}
 where the first equality follows from Equation~(2.9) in \citet{fleming1991counting} and the second equality follows from the definition of $\Lambda_C$. 
 Second, note that
 \begin{align*}
  \int_{(t,\infty)} \frac{Y^{\dagger}(u) \,\d \Lambda_C(u)}{K(u)}
  & = \Delta \int_{(t,X)} \frac{\d \Lambda_C(u)}{K(u)} + (1-\Delta) \int_{(t,X]} \frac{\d \Lambda_C(u)}{K(u)} \\
  & = \Delta \int_{(t,X)} \d K^{-1}(u) + (1-\Delta) \int_{(t,X]} \d K^{-1}(u) \\
  & = \Delta I(X>t) \Bigl\{ K^{-1}(X-) - K^{-1}(t) \Bigr\} - (1-\Delta) I(X>t) \Bigl\{ K^{-1}(X) - K^{-1}(t) \Bigr\} \\
  & = I(X>t) \left\{ \frac{\Delta}{K(X-)} + \frac{1-\Delta}{K(X)} - \frac{1}{K(t)} \right\}. 
 \end{align*}
 Now, the identity in \eqref{eq:rr1} follows immediately from
 \begin{align*}
  \int_{(t,\infty)} \frac{\d M_C(u)}{K(u)}
  & = \frac{(1-\Delta)I(X>t)}{K(X)} - \int_{(t,\infty)} \frac{Y^{\dagger}(u) \,\d \Lambda_C(u)}{K(u)} \\
  & = I(X>t) \left\{ \frac{1}{K(t)} - \frac{\Delta}{K(X-)} \right\},
 \end{align*}
and the identity in \eqref{eq:rr2} follows from the identity in \eqref{eq:rr1} by 
letting $t \to 0+$ from the right.
\end{proof}

\section{Coarsening At Random: Proof of Theorem \ref{thm:car_caus_surv}}

\begin{proof}

Recall that the full data is $(L,T^*_0, T^*_1)$ 
and the observed data is $(L, A, \Delta, X)$. The coarsening variable 
$(A, C^*_0, C^*_1)$ determines the observed data through the mapping
\begin{equation*}
    \Phi(L,T^*_0, T^*_1; A, C^*_0, C^*_1) = (L, A, \Delta=I(T^*_A \leq C^*_A), 
    X = T^*_A \wedge C^*_A). 
\end{equation*}

The following lemma plays a role in the remainder of the proof.
\begin{lem}
 \label{lem:X}
 The set of full data values that could produce an observed data value is
 \begin{equation*}
     \mathcal{X}(l, a, \delta, x)
     = \{(l', t_0, t_1) \, : \, 
     l'=l, 
     t_a = x \text{ if } \delta=1 \text{ and } t_a > x \text{ if } \delta=0
     \}.
 \end{equation*}
 Note that $t_{1-a}$ is arbitrary.
\end{lem}

Now we characterize coarsening at random, or CAR, defined as
\begin{align*}
     [(L,A,\Delta,X)=(l,a,\delta,x)& \mid (L,T^*_0, T^*_1)=(l', t_0, t_1)] =  \\
     &[(L,A,\Delta,X)=(l,a,\delta,x) \mid (L,T^*_0, T^*_1)=(l'', t_0', t_1')]
\end{align*}
for all $(l', t_0, t_1), (l'', t_0', t_1') \in \mathcal{X}(l,a,\delta,x),$
where $[\cdot]$ denotes densities and $(l,a,\delta,x)$ is an arbitrary observation \citep{van2003unified}. 
The condition roughly expresses that the observed data depends on the full data only 
through the part that is observed. Note that we may work with the density as both the 
Lebesgue and counting measures satisfy CAR themselves \citep[page 25]{van2003unified}.

 Throughout we use the notation $[\cdot]$ to denote a density and rely on the following decomposition:
  \begin{align}
    [L=l, A=a, & \Delta=\delta, X=x \mid L=l',  T^*_0=t_0, T^*_1=t_1] 
    \nonumber \\
    =\, & I(l=l') [A=a, \Delta=\delta, X=x \mid L=l, T^*_0=t_0, T^*_1=t_1] \nonumber \\
    =\, & I(l=l') [\Delta=\delta, X=x \mid A=a, L=l, T^*_0=t_0, T^*_1=t_1] \nonumber \\
    & \hspace{10mm} \times [A=a \mid L=l, T^*_0=t_0, T^*_1=t_1] \nonumber \\
    =\, & I(l=l') [\Delta^*_a=\delta, \delta T_a^* + (1-\delta) C_a^* = x \mid A=a, L=l, T^*_0=t_0, T^*_1=t_1] \nonumber \\
    & \hspace{10mm} \times [A=a \mid L=l, T^*_0=t_0, T^*_1=t_1], \label{eq:dens}.
 \end{align}
 
 We start by establishing the forward direction. Assume coarsening at random. 
 Our proof strategy to find necessary conditions is to show that certain conditional 
 densities do not depend on part of their conditioning set. 
 
 We first show that $A \ind (T_0^*, T_1^*) \mid L$. For any $l', t_0, t_1$, we calculate that
 \begin{align*}
     [A=a \mid L=l', T^*_a = t_a ] 
    =\, & \int [L=l, A=a, X=x, \Delta=\delta \mid L=l', T^*_a = t_a] \d\nu(x,\delta,l) \\
    =\, & \int [L=l, A=a, X=x, \Delta=\delta \mid L=l', T^*_a = t_a] \d\nu(x,\delta,l) \\
    \stackrel{\mathrm{CAR}}{=}\, & \int [L=l, A=a, X=x, \Delta=\delta \mid L=l', T^*_0 = t_0, T^*_1 = t_1] \d\nu(x,\delta,l) \\
    =\, & \int [L=l, A=a, X=x, \Delta=\delta \mid L=l', T^*_0 = t_0, T^*_1 = t_1] \d\nu(x,\delta,l) \\
    =\, & [A=a \mid L=l', T^*_0 = t_0, T^*_1 = t_1]
 \end{align*}
 where the integral to be with respect to the appropriate measure. Since the latter conditional probability sums across $a$ to unity, also
 \begin{align*}
     [A=1 \mid L=l', T^*_1 = t_1] + [A=0 \mid L=l', T^*_0 = t_0] = 1
 \end{align*}
 for all $t_0, t_1$.
 Since the identity holds for $(t_0, t_1)$ and for $(t_0', t_1')$, we have that
 \begin{align*}
     & [A=1 \mid L=l', T^*_1 = t_1] - [A=1 \mid L=l', T^*_1 = t_1'] \\
     =\, & [A=0 \mid L=l', T^*_0 = t_0'] -  [A=0 \mid L=l', T^*_0 = t_0]
 \end{align*}
 for any $(t_0, t_1), (t_0', t_1')$. Setting $t_0=t_0',$ we find that
 \begin{equation*}
     [A=1 \mid L=l', T^*_1 = t_1] - [A=1 \mid L=l', T^*_1 = t_1'] = 0
 \end{equation*}
 for all $t_1,t_1'$. Therefore, the density does not depend on the value of $T_1^*$, and we've shown that
 \begin{align*}
     [A=1 \mid L=l', T_1^*=t_1, T_0^* = t_0]
     & = [A=1 \mid L=l', T_1^*=t_1] \\
     & = [A=1 \mid L=l'],
 \end{align*}
 for all $t_0, t_1$. 
 Since $A$ is binary and this density describes its full distribution, 
 we've shown that $A \ind (T_0^*, T_1^*) \mid L$.
 
 Next, we show that $T^*_a \ind C^*_a \mid A=a, L$ on $T^*_a > C^*_a$. Consider $\delta = 0$. By Equation~\eqref{eq:dens},
 \begin{align*}
    [L=l, A=a, & \Delta=0, X=x \mid L=l', T^*_0=t_0, T^*_1=t_1]  =\,  \\
    & [L=l, \Delta=0, X=x \mid L=l', A=a, T^*_a=t_a, T^*_{1-a}=t_{1-a}] \\
    & \hspace{10mm} \times [A=a \mid L=l', T^*_0=t_0, T^*_1=t_1] \\
    =\, & I(l=l') [\Delta_a^* = 0, C_a^*=x \mid A=a, L=l, T_0^* = t_0, T^*_1=t_1] [A=a \mid L=l] \\
    =\, & I(l=l', t_a>x) [C_a^*=x \mid A=a, L=l, T_a^* = t_a, T^*_{1-a}=t_{1-a}] [A=a \mid L=l],
 \end{align*}
 since $C_a^*=x$ and $T_a^*=t_a$ implies $\Delta_a^* = 0$ when $t_a > x$. 
 Because these values are equal for $t_a, t_a'>x$ and any $t_{1-a}, t_{1-a}'$ (by CAR), we find
 \begin{align*}
     [C_a^*=x \mid A=a, L=l, T_a^* = t_a, T^*_{1-a}=t_{1-a}] 
     =\, & [C_a^*=x \mid A=a, L=l, T_a^* = t_a', T^*_{1-a}=t_{1-a}'],
 \end{align*}
 after dropping the propensity score $[A=a \mid L=l]$. 
 Therefore $(T_0^*, T_1^*) \ind C_a^* \mid A=a, L$ on $T_a^* > C_a^*$. 
 
 Finally, we show that $T_{1-a}^* \ind \Delta_a^* \mid A=a, L, T^*_a$. 
 Consider $\delta = 1$. By Equation~\eqref{eq:dens},
 \begin{align*}
    [L= & l, A=a,  \Delta=1, X=x \mid L=l',  T^*_0=t_0, T^*_1=t_1] 
      \\
  &  =\, [L=l, \Delta=1, X=x \mid L=l', A=a, T^*_a=t_a, T^*_{1-a}=t_{1-a}] \\
    & \hspace{10mm} 
    \times [A=a \mid L=l, T^*_0=t_0, T^*_1=t_1] \\
   &  =\,  I(l=l') [\Delta_a^* = 1, T_a^* = x \mid A=a, L=l, T_a^*=t_a, T_{1-a}^* = t_{1-a}] [A=a \mid L=l]\\
   &  =\,  I(l=l', t_a = x) [\Delta_a^* = 1 \mid A=a, L=l, T_a^*=x, T_{1-a}^* = t_{1-a}] [A=a \mid L=l].
 \end{align*}
 By dropping the propensity score term $[A=a \mid L=l]$ and setting the remaining conditional density equal 
 for $t_a,t_a' = x$ and any $t_{1-a},t_{1-a}'$, we find that for all $a,l,x$,
 \begin{align*}
      [\Delta_a^* = 1 \mid & A=a, L=l, T_a^*=x, T_{1-a}^* = t_{1-a}] \\
     =\, & [\Delta_a^* = 1 \mid A=a, L=l, T_a^*=x, T_{1-a}^* = t_{1-a}'];
 \end{align*}
 therefore $(T_{1-a}^*) \ind \Delta_a^* \mid A=a, L, T^*_a$, as $\Delta^*_a$ is dichotomous.

 Now, we establish the reverse direction. Assume the conditional independence statements. 
 To establish coarsening at random, we use Equation~\eqref{eq:dens} to study the densities
 \begin{align*}
     [L & =l, A=a,  \Delta=\delta, X=x \mid L=l', T^*_0=t_0, T^*_1=t_1] = \\
&     =\,  [L=l, \Delta=\delta, X=x \mid L=l', A=a, T^*_a=t_a, T^*_{1-a}=t_{1-a}] \\
    & \hspace{10mm} 
    \times [A=a \mid L=l, T^*_0=t_0, T^*_1=t_1] \\
 &   =\,  I(l=l') [\Delta_a^*=\delta, X^*_a = x \mid A=a, L=l, T^*_a=t_a, T^*_{1-a}=t_{1-a}] [A=a \mid L=l].
 \end{align*}
 When $\delta=0$, the first conditional density
 \begin{align*}
     [\Delta^*_a=0, & C_a^*=x \mid A=a, L=l, T^*_a=t_a, T^*_{1-a}=t_{1-a}]  \\
    =\, & I(t_a>x) [C_a^*=x \mid A=a, L=l, T^*_a=t_a, T^*_{1-a}=t_{1-a}] \\
    =\, & I(t_a>x) [C_a^*=x \mid A=a, L=l, T^*_a=t_a] \\
    =\, & I(t_a>x) [C_a^*=x \mid A=a, L=l, T^*_a=x+1],
 \end{align*}
 since the density does not depend on $t_a$ as long as $t_a>x$. Note, we write $T_a^* = x+1$ in the conditioning set, but we could have equivalently written $T_a^* = x+\epsilon$ for any $\epsilon>0$. 
 When $\delta=1$, the same term 
 \begin{align*}
      [\Delta_a^*=1, & T_a^*=x \mid A=a, L=l, T^*_a=t_a, T^*_{1-a}=t_{1-a}] \\
     =\, & I(t_a=x) [\Delta_a^* = 1 \mid A=a, L=l, T_a^* = x, T_{1-a}^* = t_{1-a}] \\
     =\, & I(t_a=x) [\Delta_a^*=1 \mid A=a, L=l, T_a^* = x].
 \end{align*}
 
 Putting these together, for any $(l', t_0, t_1) \in \mathcal{X}(l, a, \delta, x)$, the density
  \begin{align*}
  [L=l  , A=a, & \Delta=\delta, X=x \mid L=l', T^*_0=t_0, T^*_1=t_1] \\
    =\, & [A=a\mid L=l] \Big\{ \delta [\Delta_a^*=1 \mid A=a, L=l, T_a^* = x] \\
    \hspace*{5mm}  & + (1-\delta) [C_a^*=x \mid A=a, L=l, T_a^* = x+1] \Big\},
 \end{align*}
 since the indicators always equal one by Lemma \ref{lem:X}. Therefore the density is a 
 measurable function of the observed data $(l, a, \delta, x)$ and hence CAR holds.
 
 
We now turn to result on sequential CAR, which in this case means that CAR holds at each stage
of coarsening. Consider the mapping from full data to observed data defined by 
\begin{align*}
    & \Phi(L,T^{\full}_0, T^{\full}_1; A, C^{\full}_0, C^{\full}_1) \\
    & \hspace{15mm} = \Bigl( L, A, \Delta=I(T^{\full}_A \leq C^{\full}_A), X=\min\{T^{\full}_A, C^{\full}_A\} \Bigr). 
\end{align*}
In what follows, it will be important to remember that all coarsening variables (here, $A$ and $C^{\full}_0, C^{\full}_1$) appear to the right of the semi-colon.

We recall that the observed data mapping satisfies CAR if and only if for each $a=0,1$, 
 \begin{align*}
     & (T^{\full}_0, T^{\full}_1) \ind A \mid L, \\
     & (T^{\full}_0, T^{\full}_1) \ind C^{\full}_a \mid A=a, L \text{ on } C^{\full}_a < T^{\full}_a, \\
     & T_{1-a}^{\full} \ind I(T_a^{\full} \leq C_a^{\full}) \mid A=a, L, T^{\full}_a.
 \end{align*}
Define the mappings
\begin{align*}
 \Phi_{\text{int}}(L,T^{\full}_0, T^{\full}_1; A) 
 & = \bigl(L, A, T^{\full}=T^{\full}_A \bigr), \\
 \Phi_{\text{obs}}(L,A,T^{\full}; C^{\full}) 
 & = \bigl( L,A,X=T^{\full}\wedge C^{\full},\Delta=I(T^{\full}\leq C^{\full}) \bigr).
\end{align*}
The mapping $\Phi_{\text{int}}(L,T^{\full}_0, T^{\full}_1; A)$ coarsens the uncensored potential failure time outcomes by $A$; this coarsened but uncensored data is then subject to right censoring, represented by $ \Phi_{\text{obs}}(L,A,T^{\full}; C^{\full}).$
In particular, we see that the observed data 
can be written as  
$\Phi_{\text{obs}}\{\Phi_{\text{int}}(L,T^{\full}_0,T^{\full}_1; A) ; C^{\full}\}.$

The data induced under 
$\Phi_{\text{int}}(L,T^{\full}_0, T^{\full}_1; A)$ is evidently contained
within that induced by $\Phi_{\text{obs}}(L,A,T^{\full}; C^{\full}),$ which itself is evidently
equivalent to that generated by $\Phi(L,T^{\full}_0, T^{\full}_1; A, C^{\full}_0, C^{\full}_1).$
We next show that both $\Phi_{\text{int}}$ and $\Phi_{\text{obs}}$ may 
each be written in terms of the full mapping $\Phi$. First, recalling that all coarsening
variables appear to the right of the semi-colon in $\Phi,$ we note that
\begin{align}
  \label{eq:car-int}
\Phi(L,T^{\full}_0, T^{\full}_1; A, \infty, \infty) = 
\bigl( L, A, \Delta=I(T^{\full}_A \leq \infty), X=\min\{T^{\full}_A, \infty \} \bigr). 
 \end{align}
 Since $T^{\full}_A < \infty$ almost surely, we see that \eqref{eq:car-int} reduces to
$ \bigl( L, A, 1, T^{\full}_A \} \bigr),$ which is equivalent to $ \bigl( L, A, T^{\full}_A \bigr) = \Phi_{\text{int}}(L,T^{\full}_0, T^{\full}_1; A).$
Second, noting again that all coarsening
variables appear to the right of the semi-colon in $\Phi,$ we have
\begin{align}
\nonumber
\Phi\bigl( (L,A), T^*, T^*; & \emptyset, C^*, C^* \bigr) \\
& = \Bigl( (L, A), \emptyset,  \Delta=I(T^{\full} \leq C^{\full}), X=\min\{T^{\full}, C^{\full}\} \Bigr);
\label{eq:car-obs}
\end{align}
however, the information in \eqref{eq:car-obs} is just
$\Phi_{\text{obs}}(L,A,T^{\full}; C^{\full}).$
%
%
Thus, because the data at each stage of coarsening has been represented through a suitable choice of $\Phi,$ we can now apply the CAR characterization for $\Phi$ stated earlier. Based on \eqref{eq:car-int}, the CAR  characterization for $\Phi_{\text{int}}$ is for each $a=0,1$,
\begin{align*}
    & (T^{\full}_0, T^{\full}_1) \ind A \mid L, \\
    & (T^{\full}_0, T^{\full}_1) \ind \infty \mid A=a, L \text{ on } \tau < T^{\full}_a, \\
    & T_{1-a}^{\full} \ind I(T_a^{\full} \leq \infty) \mid A=a, L, T^{\full}_a.
\end{align*}
The conditions in the second and third lines are degenerate because 
$T^{\full}_a \leq \infty$ always, so the CAR characterization for $\Phi_{\text{int}}$ is simply 
\begin{equation}
    (T^{\full}_0, T^{\full}_1) \ind A \mid L. \label{eq:seq-car1}
\end{equation}
This condition is both necessary and sufficient for CAR at this stage of coarsening.
Based on \eqref{eq:car-obs}, the CAR characterization for $\Phi_{\text{int}}$ is for each $a=0,1$
\begin{align*}
     & (T^{\full}, T^{\full}) \ind \emptyset \mid (L,A), \\
     & (T^{\full}, T^{\full}) \ind C^{\full} \mid  (L,A) \text{ on } C^{\full} < T^{\full}, \\
     & T^{\full} \ind I(T^{\full} \leq C^{\full}) \mid  (L,A), T^{\full}.
\end{align*}
The condition that ``$A=a$'' does not appear in the second and third lines
since $A$ does not play the role of a coarsening variable here; moreover,
the conditions in the first and third line are also degenerate, so the 
CAR characterization for $\Phi_{\text{int}}$ reduces to
$T^{\full} \ind C^{\full} \mid L, A \text{ on } C^{\full} < T^{\full}$, 
which can be equivalently expressed as 
\begin{equation}
    T^{\full}_a \ind C^{\full}_a \mid L, A=a \text{ on } C^{\full}_a < T^{\full}_a
    \mbox{ for each $a=0,1$}. \label{eq:seq-car2}
\end{equation}
Similarly to the first stage, this condition is also necessary and sufficient 
for CAR at this stage of coarsening.  The proof is now complete, as
the stated conditions \eqref{eq:seq-car1} and \eqref{eq:seq-car2} are those
given in the statement of the theorem, and are necessary and sufficient for CAR
at each stage.

\end{proof}

\section{Observed Data Identification}

\begin{proof}[Lemma \ref{lem:surv-ident}]
    Using results in \citet[Theorem 11]{gill1990survey}, 
    \begin{align*}
        K^{\full}(s\mid a,l)
        &= \Pf (C^*_a > s \mid A=a, L=l) \\
        &= \Prodi_{u\in (0,s]} 
        \left\{ 1 - \d \Lambda_C^{\full}(u \mid a, l)
        \right\},
    \end{align*}
    where $\prodi$ is the product integral and $\Lambda_C^{\full}$ is the cumulative hazard. The cumulative hazard may be identified as
    \begin{align*}
        \Lambda_C(t; a, l)
        & = \int_{(0, t]} \frac{\d \P(X \leq u, \Delta=0 \mid A=a, L=l)}{\P(X > u, \Delta=1 \text{ or } X \geq u, \Delta=0 \mid A=a, L=l)} \\
        & = \int_{(0, t]} \frac{\d \P(X \leq u, \Delta=0 \mid A=a, L=l)}{\Pf(T^*_a > u, C^*_a \geq u \mid A=a, L=l)} \\
        & \stackrel{(eq. 1)}{=} \int_{(0, t]} \frac{\Pf(T^*_a > u \mid A=a, L=l) \d \Pf(C^*_a \leq u \mid A=a, L=l)}{\Pf(T^*_a > u, C^*_a \geq u \mid A=a, L=l)} \\
        & \stackrel{(eq. 2)}{=} \int_{(0, t]} \frac{\d \Pf(C^*_a \leq u \mid A=a, L=l)}{\Pf(C^*_a \geq u \mid A=a, L=l)} \\
        & = \Lambda_C^{\full}(t; a, l).
    \end{align*}

    We now justify equalities (eq.\ 1) and (eq.\ 2). The first equality follows from
    \begin{align*}
         \P(X\leq t, \Delta=0 & \mid A=a, L=l) 
        =\,  \Pf(C^*_a \leq t, C^*_a < T^*_a \mid A=a, L=l) \\
        =\, & \int_{(0, t]} \Pf(T^*_a > u \mid A=a, L=l) \d \Pf(C^*_a \leq u \mid A=a, L=l),
    \end{align*}
    making use of the condition that $T^*_a \ind C^*_a | A = a, L = l$
    on $T^*_a > C^*_a.$ The second equality follows from
    \begin{align*}
        \frac{\Pf(T^*_a > u, C^*_a \geq u \mid A=a, L=l)}{\Pf(T^*_a > u \mid A=a, L=l)} 
        & =  \Pf(C^*_a \geq u \mid T^*_a > u, A=a, L=l) \\
        & =  1-\Pf(C^*_a < u \mid T^*_a > u, A=a, L=l) \\
        & =  1-\Pf(C^*_a < u \mid A=a, L=l) \\
        & = \Pf(C^*_a \geq u \mid A=a, L=l),
    \end{align*}
    the third equality again making use of the indicated independence.

    Finally, the result follows by plugging into the product integral above. That $H^* = H$ follows analogously. 
\end{proof}

\begin{proof}[Theorem \ref{prop:identif}]
    Notice that
    \begin{align*}
        \E \{ \varphi_{\eta, a} (t; \mO; \P) \}
        & = \E \left\{ \frac{I(A=a)}{\pi(A; L)} \frac{\Delta}{K(X-; A, L)} I(X > t) \right\} \\
        & = \E^* \left\{ \frac{I(A=a)}{\pi(a; L)} \frac{I(T^*_a \leq C^*_a)}{K^*(T^*_a-; a, L)} I(T^*_a > t) \right\} \\
        & = \E^* \left\{ \frac{I(A=a)}{\pi(a; L)} \frac{I(T^*_a > t)}{K^*(T^*_a-; a, L)} \P^* \left( T^*_a \leq C^*_a \middle| A=a, L, T^*_a \right) \right\} \\
        & = \E^* \left[ \frac{I(A=a)}{\pi(a; L)} \frac{I(T^*_a > t)}{K^*(T^*_a-; a, L)} \Bigl\{ 1 - \P^* \left( C^*_a < T^*_a \middle| A=a, L, T^*_a\right) \Bigr\} \right].
    \end{align*}

The expression $\P^* \left( T^*_a \leq C^*_a \mid A=a, L, T^*_a \right)$ in the above is rewritten so that the CAR assumption, which furnishes conditional independence ``on $C^*_a < T^*_a$,'' may be leveraged. In particular, under this assumption, we can write
\begin{align*}
\P^* \bigl( C^*_a < T^*_a  & \mid A=a, L, T^*_a \bigr) 
= \P^* \left( C^*_a < u \mid A=a, L \right) \bigl|_{u= T^*_a} \\
& = 
\bigl\{ 1- \P^* \left( C^*_a \geq u \mid A=a, L\right)  \bigr \} \bigl|_{u= T^*_a}
= 1 - K^*(T^*_a-; a, L);
\end{align*}
substitution into the last expression  given for
$ \E \{ \varphi_{\eta, a} (t; \mO; \P) \}$
then leads to 
\begin{align*}
  \E \{ \varphi_{\eta, a} & (t; \mO; \P) \} 
          = \E^* \left\{ \frac{I(A=a)}{\pi(a; L)} I(T^*_a > t) \right\}  \\
         &  = \E^* \left\{ \frac{I(T^*_a > t)}{\pi(a; L)} \P^* \left( A=a \middle|  L \right) \right\} \\
         & = \E^* \left\{ I(T^*_a > t) \right\},
\end{align*}
proving the desired identification result.
\end{proof}

\section{On other influence functions for $\eta_a(t)$}
\label{sec:litreview:surv:class}

\subsection{A Missing Data Derivation of 
\eqref{eq:tsiatisclass}}
\label{sec:class:tsiatis}

In this section, we derive the class of influence functions reported in \citet{anstrom2001utilizing,bai2013doubly,bai2017optimal} under a particular missing data model, rather than the causal model considered in this paper. 

Let the full data as $(L,T^*),$ and denote the full data estimand as $\psi(t) := \PP^*(T^*>t)$ for some $t>0$. In addition, define $C^*$ to be a potential censoring time. As will be shown below, \citet{anstrom2001utilizing,bai2013doubly,bai2017optimal} essentially derive their class using results in \citet{tsiatis2006semiparametric} by assuming that  
$T^*_1 = T^*$ and that $C^*_1 = C^*;$ critically, 
$(T^*_0,C^*_0)$ is implicitly considered to be missing (or, just $T^*_0$, if we additionally assume $C^*_1 = C^*_0$). 

We begin by defining $R \in \{0,1\}$ to be a binary indicator variable, and assuming that the observed data arises from the mapping
\begin{equation*}
    \Phi(L, T^*; R,  C^*) = (L, R, R\min\{T^*,  C^*\}, R I(T^*\leq C^*)) = (L, R, R X, R \Delta),
\end{equation*}
where $X = T^* \wedge C^*$ and $\Delta=I(T^*\leq  C^*).$ Evidently, $R = 1$ implies that  $(L,X,\Delta)$ is observed, 
whereas $R = 0$ leads only to the observation of $L.$  Following \citet{tsiatis2006semiparametric}, this observed data construction 
amounts to using a coarsening variable $\mC \in(0,\infty)$ that is given by 
\begin{equation*}
    \mC = 0 I(R=0) + X I(R=1, \Delta = 0) + \infty I(R=1, \Delta=1).
\end{equation*}
The variable $\mC$ cleary reflects the indicated coarsening pattern and its structure reveals that the coarsening reflected
in $\Phi$ is monotone.

Note that the construction above does not explicitly assume the existence of a treatment variable. However, upon equating $R$ with $A$, it
can be seen that the above construction amounts to considering the problem of estimating $\psi(t) = P\{ T^* > t\} = P\{ T^*_1 > t\}$ from the subset of observations where $A=1$ only; when $A = 0$, the information for estimating $\psi(t)$ is considered to be missing.

In deriving the corresponding class of influence functions, we assume that $\PP^*(R=1, T^* \leq  C^* \mid L, T^*) > 0,$ and 
also that the following version of sequential CAR also holds: 
\begin{align*}
    & R \ind T^* \mid L, \\
    &  C^* \ind T^* \mid L, R=1 \text{ on } C^* < T^*.
\end{align*}
It follows that the coarsening probabilities simplify as $\PP^*(R=1\mid T^*, L) = \pi(1; L)$ and $\PP^*(C^* > u \mid T^*, L, R=1) = K(u; L)$ and hence $\PP^*(R=1,\Delta=1 \mid T^*, L) = \pi(1; L) K(X-; L)$ when $R,\Delta=1$.

For simplicity, assume both $\pi(1;L)$ and $K(t;L)$ are known. By \citet[Theorems 7.2 and 8.3]{tsiatis2006semiparametric}, the 
class of observed-data influence functions under the indicated monotone coarsening mechanism is given by
\begin{equation}
\label{Tsiatis R class}
    \frac{I(R=1, \Delta=1)}{\pi(1;L) K(X-; L)} \{I(T^*>t) - \psi \} + \mathrm{AS},
\end{equation}
where $\mathrm{AS}$ is the ``augmentation space'', 
since $I(T^*>t) - \psi$ is the unique full data influence function. 

Before characterizing the augmentation space, we define
\begin{align*}
    \lambda_u & = \PP^*(\mC=u \mid \mC \geq u, L, T^*), \\
    K_u & = \PP^*(\mC>u \mid L, T^*).
\end{align*}
When $u=0$, these simplify as 
\begin{align*}
    \lambda_0 &= \PP^*(R=0\mid T^*,L) = 1-\pi(1; L), \\
    K_0 &= \pi(1; L).
\end{align*}
When $u>0$ is finite, 
these expressions simplify as
\begin{align*}
    \lambda_u &= I(T^*>u) \PP^*( C^* > u \mid  C^* \geq u, T^* > u, R=1, L); \\ 
    K_u &     = \PP^*( C^* > u \mid L, R=1) \pi(1; L).
\end{align*}

By \citet[Theorem 9.2]{tsiatis2006semiparametric}, a typical element of the augmentation space may be written as
\begin{align*}
    & \frac{I(\mC=0) - \lambda_0 I(\mC \geq 0)}{K_0} \tilde h_1(L) + \int_0^\infty \frac{I(\mC=u) - I(\mC \geq u) \lambda_u}{K_u} \tilde h_2(u; L) \\
    =\, &  \frac{(1-R) - \{1-\pi(1; L)\}}{\pi(1; L)} \tilde h_1(L) + \int_0^\infty \frac{R \,\d N_C(u) - R Y^{\dagger}(u) \lambda_u}{\pi(1; L)\PP^*( C^* > u \mid R=1, L)} \tilde h_2(u; L) \\
    =\, &  - \frac{R-\pi(1; L)}{\pi(1; L)} \tilde h_1(L) + \frac{R}{\pi(1; L)} \int_0^\infty \frac{\d M_C(u; 1,L)}{\PP^*( C^* > u \mid R=1, L)} \tilde h_2(u; L).
\end{align*}
Substituting this into \eqref{Tsiatis R class} and treating $R$ as being equivalent to $A,$ we obtain
\begin{align*}
    \frac{I(A=1)}{\pi(1;L)} \frac{\Delta}{K(X-; L)} & \{I(T^*>t) - \psi(t) \} -  \frac{I(A=1)-\pi(1; L)}{\pi(1; L)} \tilde h_1(L) \\
    & + \frac{I(A=1)}{\pi(1; L)} \int_0^\infty \frac{\d M_C(u; 1,L)}{\PP^*( C^* > u \mid A=1, L)} \tilde h_2(u; L).
\end{align*}
Since $\tilde h_1(L)$ is arbitrary, we can use the second identity in Lemma \ref{lem:rr-lems} to algebraically rewrite the above expression as 
\begin{align*}
    \frac{I(A=1)}{\pi(1;L)} \frac{\Delta}{K(X-; L)} & I(T^*>t)  - \psi(t) - \frac{I(A=1)-\pi(1; L)}{\pi(1; L)} h_1(L) \\ & + \frac{I(A=1)}{\pi(1; L)} \int_0^\infty \frac{\d M_C(u; 1,L)}{\PP^*( C^* > u \mid A=1, L)} h_2(u; L),
\end{align*}
where $h_1(L) = \tilde h_1(L)+\psi(t)$ and $h_2(u;L) =  \tilde  h_2(u; L) + \psi(t)$. It is not difficult to see that this last equation is simply 
\eqref{eq:tsiatisclass} with $a = 1$, hence the same class as that given in \citet{bai2013doubly}.

\subsection{Proof of Theorem \ref{thm: EIF equal}}
\label{sec:vm:surv_equiv}

In this section, we prove that two forms of the efficient influence function for 
$P(T_a > t)$ that appear in the literature do in fact agree. We stress that the result 
applies in either continuous or discrete time or a mixture of the two.

For notational simplicity, we suppress the conditioning set $A, L$. We start by simplifying the uncentered efficient influence function representation
\begin{align*}
H(t) - & \int_{(0,t]} \frac{I(A=a)}{\pi(a; L)} \frac{H(t)}{H(u)}  \frac{\d N_T(u) - Y(u) \d\Lambda_T(u)}{K(u-)} \\
    =\, & H(t) 
    - \frac{I(A=a)}{\pi(a; L)} H(t) \bigg\{ \int_{(0,t]} \frac{\d I(\Delta=1, X \leq u)}{K(u-) H(u)} - \int_{(0,t]} \frac{I(X \geq u) \d\Lambda_T(u)}{K(u-) H(u)} \bigg\}.  
\end{align*}
The first integral simplifies as
\begin{equation*}
    \int_{(0,t]} \frac{\d I(\Delta=1, X \leq u)}{K(u-) H(u)}
    = \frac{\Delta I(X \leq t)}{K(X-) H(X)},
\end{equation*}
which can be seen by splitting the integration domain into three parts: $(0,X), \{X\},$ and $(X,t]$.
We now expand the remaining integral using integration by parts. The sum
\begin{align*}
 \int_{(0,t]} \frac{Y(u)}{K(u-)} \d & \left\{ \frac{1}{H(u)} \right\} + \int_{(0,t]} \frac{Y^{\dagger}(u)}{H(u)} \d \left\{ \frac{1}{K(u)} \right\} \\
    =\, & I(X > t) \left[ \int_{(0,t]} \frac{1}{K(u-)} \d \left\{ \frac{1}{H(u)} \right\} + \int_{(0,t]} \frac{1}{H(u)} \d \left\{ \frac{1}{K(u)} \right\} \right] \\
    & \hspace{5mm} + I(X \leq t) \biggl[ \int_{(0,X]} \frac{1}{K(u-)} \d \left\{ \frac{1}{H(u)} \right\} + 
    \Delta \int_{(0,X)} \frac{1}{H(u)} \d \left\{ \frac{1}{K(u)} \right\} \\
    & \hspace{5mm} + (1-\Delta) \int_{(0,X]} \frac{1}{H(u)} \d \left\{ \frac{1}{K(u)} \right\} \biggr] \\
    =\, & I(X > t) \left[ \int_{(0,t]} \frac{1}{K(u-)} \d \left\{ \frac{1}{H(u)} \right\} + \int_{(0,t]} \frac{1}{H(u)} \d \left\{ \frac{1}{K(u)} \right\} \right] \\
     & \hspace{5mm} + I(X \leq t) \biggl[ \int_{(0,X]} \frac{1}{K(u-)} \d \left\{ \frac{1}{H(u)} \right\} 
    + \int_{(0,X]} \frac{1}{H(u)} \d \left\{ \frac{1}{K(u)} \right\} \\
    & \hspace{5mm} - \Delta \int_{\{X\}} \frac{1}{H(u)} \d \left\{ \frac{1}{K(u)} \right\} \biggr].
\end{align*}
Now, the standard integration by parts formula \citep[Theorem A.1.2]{fleming1991counting} applies since $K$ and $H$ are both right-continuous yielding for any $v>0$ that
\begin{align*}
    \int_{(0,v]} \frac{1}{K(u-)} \d \left\{ \frac{1}{H(u)} \right\} + \int_{(0,v]} \frac{1}{H(u)} \d \left\{ \frac{1}{K(u)} \right\}
    = \frac{1}{K(v)H(v)} - 1.
\end{align*}
Applying this result for $v=X$ and $v=t$ continues the simplification and shows
\begin{align*}
\int_{(0,t]} \frac{Y(u)}{K(u-)} \d & \left\{ \frac{1}{H(u)} \right\} + \int_{(0,t]} \frac{Y^{\dagger}(u)}{H(u)} \d \left\{ \frac{1}{K(u)} \right\} \\
    =\, & I(X > t) \left[ \frac{1}{K(t)H(t)} - 1 \right] + I(X \leq t) \times \\
    & \left[ \frac{1}{K(X)H(X)} - 1 - \Delta \int_{\{X\}} \frac{1}{H(u)} \d \left\{ \frac{1}{K(u)} \right\} \right] \\
    =\, & \frac{1}{K(X\wedge t)H(X \wedge t)} - 1 - \frac{\Delta I(X \leq t)}{H(X)} \left\{ \frac{1}{K(X)} - \frac{1}{K(X-)} \right\}.
\end{align*}
Now that both integrals are simplified, we plug them back into the efficient influence function and find that
\begin{align*}
 H(t) - \int_{(0,t]} \frac{I(A=a)}{\pi(a; L)} & \frac{H(t)}{H(u)} \frac{d N_T(u) - Y(u) \d\Lambda_T(u)}{K(u-)} \\
    =\, & H(t) 
    - \frac{I(A=a)}{\pi(a; L)} H(t) \bigg[ - \frac{1}{K(X\wedge t) H(X \wedge t)} + 1 \\
    & + \frac{\Delta I(X \leq t)}{H(X)} \frac{1}{K(X)} + \int_{(0,t]} \frac{Y^{\dagger}(u)}{H(u)} \d \left\{ \frac{1}{K(u)} \right\} \bigg].
\end{align*}
Rearranging terms, we find that the influence function equals
\begin{align}
\nonumber
    \frac{I(A=a)}{\pi(a; L)} & \frac{I(X>t)}{K(t)} 
    - \frac{I(A=a) - \pi(a; L)}{\pi(a; L)} H(t) + \frac{I(A=a)}{\pi(a; L)} H(t) \times \\
\label{eq:prelim-IF}
   &
    \bigg[ \frac{(1-\Delta)I(X\leq t)}{K(X)H(X)} - \int_{(0,t]} \frac{Y^{\dagger}(u)}{H(u)} \d \left\{ \frac{1}{K(u)} \right\} \bigg]. 
\end{align}
However, since
\[
- \int_{(0,t]} \frac{Y^{\dagger}(u)}{H(u)} \d \left\{ \frac{1}{K(u)} \right\}
= - \int_{(0,t]} \frac{Y^{\dagger}(u)}{H(u) K(u)} \d \Lambda_C(u),
\]
we can write
\[
\frac{(1-\Delta)I(X\leq t)}{K(X)H(X)} - \int_{(0,t]} \frac{Y^{\dagger}(u)}{H(u)} \d \left\{ \frac{1}{K(u)} \right\}
= \int_{(0,t]} \frac{1}{H(u) K(u)} d M_C(u),
\]
where $dM_C(u) = d N_C(u) - Y^\dag(u) d \Lambda_C(u).$
To complete the proof, we note that Lemma \ref{lem:rr-lems},
specifically \eqref{eq:rr1}, gives
\[
   \frac{I(X > t)}{K(t)} 
     = \frac{\Delta}{K(X-)} I(X > t) + \int_{(t,\infty)} \frac{\d M_C(u)}{K(u)}
\]
and hence that \eqref{eq:prelim-IF} can be written
\begin{align*}
    \frac{I(A=a)}{\pi(a; L)} & \frac{\Delta}{K(X-)} I(X > t)
    - \frac{I(A=a) - \pi(a; L)}{\pi(a; L)} H(t) \\
    &
    + \frac{I(A=a)}{\pi(a; L)} \bigg[ 
    \int_{(0,t]} \frac{H(t) }{H(u) K(u)} d M_C(u)
    + \int_{(t,\infty)} \frac{\d M_C(u)}{K(u)}\bigg]. 
\end{align*}
Since
$H(t \vee u) = H(t) I\{ u \leq t \} +
H(u) I\{ u > t \},$
the term in the square brackets reduces to
\[
\int_{(0,\infty)} \frac{H(t \vee u)}{H(u)} \frac{d M_C(u)}{K(u)},
\]
proving the desired equivalence.

\end{appendices}

\end{document}